\documentclass[12pt]{elsarticle}
\usepackage{amsmath}
\usepackage{amsthm}
\usepackage{amssymb}
\usepackage{amsfonts}
\usepackage{graphics} 
\usepackage{url} 
\usepackage{color} 

\usepackage{hyperref}
\hypersetup{
    colorlinks=true,
    linkcolor=blue,
    filecolor=magenta,      
    urlcolor=cyan,
}

\newtheorem{theorem}{Theorem}
\newtheorem{lemma}[theorem]{Lemma}
\newtheorem{corollary}[theorem]{Corollary}

\newtheorem{proposition}[theorem]{Proposition}
\newtheorem{definition}[theorem]{Definition}

\newtheorem{example}[theorem]{Example}
\newtheorem{remark}[theorem]{Remark}

\newtheorem{question}[theorem]{Question}

\def\F{\mathbb{F}}

\def\N{\mathbb{N}}

\begin{document}

\begin{frontmatter}

\title{Local permutation polynomials and the action of e-Klenian groups}
\author{Jaime Gutierrez}
\address{Departamento de Matem\'atica Aplicada y Ciencias de la Computaci\'on\\
Universidad de Cantabria\\
Santander, Spain}
\ead{jaime.gutierrez@unican.es}
\author{Jorge Jim\'enez Urroz}
\address{Departamento de Matem\'aticas \\
Universitat Polit\'ecnica Catalunya\\
Barcelona, Spain}
\ead{jorge.urroz@upc.edu}

\begin{abstract}
Permutation polynomials of finite fields have many applications in Coding Theory, Cryptography and Combinatorics. 
In the first part of  this paper we present a new family of local permutation polynomials based  on a  class of  symmetric subgroups without fixed points, the so called  e-Klenian groups.   In the second part we use the fact that bivariate local permutation polynomials define Latin Squares, to 
discuss several constructions   of Mutually Orthogonal Latin Squares (MOLS)   and, in particular, we provide a new family of MOLS  on size a prime power.

\end{abstract}

\end{frontmatter}

\section{Introduction}

 Let $q$ be a power of prime $p$, $\F_q$ be the finite field with $q$ elements and $\F_q^n$ denote the cartesian product of $n$ copies of $\F_q$, for any integer $n\geq 1$. Also let us use the notation $\overline x=(x_1,\dots,x_n)$ and $\overline x_i=(x_1,\dots,x_{i-1},x_{i+1},\dots x_n)$. The ring of polynomials in $n$ variables over $\F_q$ will be denoted by $\F_q[\overline x]$. It is well known that any map from $\F_q^n$ to $\F_q$ can be uniquely represented as $f \in \F_q[\overline x]$ such that $\deg_{x_i}(f) < q$ for all $i=1,\ldots,n$, where $\deg_{x_i}(f) $ is the degree of $f$ as a polynomial in the variable $x_i$ with coefficients in the polynomial ring $\F_q[\overline x_i]$, see \cite{LN}. Throughout this paper, we identify all functions $\F_q^n \to \F_q$ with  such polynomials, and every polynomial, will be of degree  $\deg_{x_i}(f) < q$, unless otherwise specified.

\

We say that a polynomial $f \in \F_q[\overline x]$ is a {\it permutation polynomial} if the equation $f(\overline x)=a$ has $q^{n-1}$ solutions in $\F_q^n$ for each $a\in \F_q$. A classification of permutation polynomials in $\F_q[\overline x]$
of degree at most two is given in \cite{N1}, see also \cite{LN} for several properties and results and the particular case $n=1$.

\

A polynomial $f \in \F_q[\overline x]$ is  called a {\it local permutation polynomial} (or LPP) if  for each $i$, $1\leq i \leq n$, the  polynomial  $f(a_1,\ldots,a_{i-1},x_i,a_{i+1}, a_ n)$ is a permutation polynomial in $\F_q[x_i]$, for all choices of
 $\overline a_i \in \F_q^{n-1}$.  Clearly any  LPP is a permutation polynomial. The opposite is not true in general. We can see that by simply considering the permutation polynomial $f(\overline x)=x_1^{q-1}+x_2$, which is not an LPP since $f(x_1,a_2,\dots, a_n)$ takes only the two values 
$a_2$ and $a_2+1$.

\

 The author of \cite{M1} and \cite{M2} gives necessary and sufficient conditions for polynomials in two and three variables to be local permutations polynomials over a prime field $\F_p$. These conditions are expressed in terms of the coefficients of the polynomial.   A recent result about degree bounds for $n$  local permutation polynomials defining a permutation of $\F_q^n$  is presented in \cite{AKT}.

\

One of the main contribution in the first part of this paper is  a general construction of a family of local permutation polynomials based  on a  class of  symmetric subgroups without fixed points, the so called  e-Klenian groups. 

\

 In the second part of the paper we are interested in Latin Squares, namely $t \times t$ matrices with entries from a set $T$ of size $t$ such that each element of $T$  occurs  exactly once in every row and every column of the matrix.

\

 It is known that every Latin square can be represented by an LPP,  $f(x,y) \in \F_q[x,y]$, (see Lemma \ref{latin}) and the relevance of this representation for the study of Latin squares (also  cubes) are described in \cite{M1} and \cite{M2}.

\

 Latin squares occur in many structures such as group multiplication tables and Cayley tables. To
be precise Latin squares are referred to as the multiplication tables of an algebraic structure called a quasigroup. 

\

Two Latin squares $L_1$ and $L_2$  of order $t$ are  orthogonal if by superimposing them one obtains all ordered pairs
$ (t_i , t_j) \in T^2$, ($i, j =1,\ldots, t)$,  and mutually orthogonal latin squares (MOLS) are sets of Latin squares that are pairwise orthogonal.
The construction of MOLS is a notoriously difficult combinatorial problem and it is one of the most studied research topics in design theory \cite{Mo}. This
interest is also due to the numerous applications that MOLS have in other fields such as
cryptography \cite{S},  coding theory and many others, see \cite{KD, MGFL, W}.
 We focus on Latin squares of prime $p$ and prime power $q=p^r$ order.
The goal of the second part of this paper is providing a big family of MOLS based on the local permutation polynomials  introduced in the previous part.

\

The remainder of the paper is structured as follows. We start with
some general properties and preliminary results on local permutation polynomials in 
 Section~\ref{prelimi}. 
Due to the one to one map between  Latin squares and  local permutation polynomials  Section~\ref{dosvariables} is consecrated to polynomials only with two variables and we  provide new families of  such  local permutation polynomials, the so called $e-$Klenian polynomials.
In Section~\ref{mols} we show  general constructions of MOLS and, in particular, one based on $e-$Klenian polynomials. We conclude with Section~\ref{conclusiones}, which makes some final 
comments and poses open questions.

\section{Elementary properties  and families of  local permutations polynomials} \label{prelimi}
 
  Our first observation in this section will be related with the degree of local permutation polynomials. For two variables,  it is shown in \cite{DHK} that the degree of a LPP in $\F_q[x_1,x_2]$ is bounded above by $2(q-2)$. The next result gives a natural generalization of this bound to several variables.
   
 \begin{proposition} \label{bound-degree} Let $n\ge 2$ be an integer. Any local permutation polynomial  $f \in \F_q[\overline x]$  is linear if $q=2$ and has degree at most  $n(q-2)$ otherwise.
 \end{proposition}
\begin{proof}
It is straightforward  if $q=2$, so let assume $q>2$ and  $\deg_{x_i}(f) < q$.
 We will prove that $\deg_{x_i}(f) < q-1$ for every variable $x_i$ for  $i=1,\ldots,n$, and for that, clearly it is enough to prove it for  $i=1$, the rest being analogous. Then, we write the polynomial $f = M_{q-1}x_1^{q-1}+ M_{q-2}x_1^{q-2}+\cdots+M_0$, such that $M_i \in \F_q[x_2,\ldots,x_n]$. Suppose that $M_{q-1}$ is  a nonzero polynomial, then there exists $(a_2,\ldots,a_n) \in \F_q^{n-1}$ such that 
$0\not=M_{q-1}(a_2,\ldots,a_n) \in \F_q$, but then $f(x_1,a_2, \ldots,a_n) \in \F_q[x_1]$ is a univariate permutation polynomial of degree $q-1$, which is a contradiction, since there is no permutation polynomial of $\F_q$ of degree a divisor of $q-1$, see \cite{LN}.
\end{proof}

Note that,  apart from the trivial case $n=1$, for $q=2$ any permutation polynomial is also a LPP, since as we have seen they are linear. 

 \

One of the main goals in the theory  is to find new  families of local permutation polynomials. The next two results can be used to construct some of them.
Suppose $ f \in \F_q[\overline x]$ is  of the form 
$$f(\overline x) = g(x_1,\ldots,x_m)  +h(x_{m+1},\ldots,x_n), \quad 1\leq m < n. $$
It is known that, If at least one of $g$ and $h$ is a permutation polynomial over $\F_q$, then $f$ is a permutation polynomial over $\F_q$,  and the inverse is also  true when $q$ is prime, see \cite{N2}. However for LPP we have the inverse for any $q$, not necessarily prime.

\begin{theorem}\label{nearseparated} Let $ f \in \F_q[\overline x]$  of the form 
$$f(\overline x) = g(x_1,\ldots,x_m)  +h(x_{m+1},\ldots,x_n), \quad 1\leq m < n $$
Then $f$ is an LPP if and only if $g$ and $h$ are local permutation polynomials.
\end{theorem}
\begin{proof}
It is immediate from the fact that any polynomial  $g$ is a permutation polynomial  if and only if  $g+a$ is also permutation polynomial, for any constant $a\in \F_q$.

\end{proof}

The following  provide another  way to construct  local permutation  polynomials.

 \begin{theorem}\label{composition} Let $ f \in \F_q[\overline x]$  be a (local) permutation polynomial. 

\medskip

\noindent 1. For any  permutation polynomial  $g(z) \in \F_q[z] $, then $g(f(\overline x))$ is a  (local) permutation polynomial.
 
\medskip

\noindent 2. Let $h_1(x_1), \ldots, h_n(x_n)$ be permutation polynomials,  then
$f(h(x_1),\ldots,h(x_n))$ is a (local) permutation polynomial.

\end{theorem}
\begin{proof}
Both of them are trivial consequence of the fact that composition of univariate permutation polynomial is again a permutation polynomial.

\end{proof}

The previous results can be used to find local permutation polynomials with the maximum degree allowed by  Proposition \ref{bound-degree}, and hence  extending the result in paper \cite{DHK}  where the authors proved that there are local  permutation polynomials in $\F_q[x,y]$ of  sharp degree $2q-4$ for $q>3$. For instance, since  $g(x)=x^3$ is a permutation polynomial in $\F_5[x]$ and, hence, also an LPP  since $n=1$, and $h(x,y,z)=x^3+y^3+z^3$ is an LPP by Theorem \ref{nearseparated}, we have that
\begin{eqnarray*}
f&=& (x^3+ y^3+z^3)^3\\
&=&x^{3} y^{3} z^{3} + 3 x^{3} y^{2} + 3 x^{2} y^{3} + 3 x^{3} z^{2} + 3y^{3} z^{2} + 3 x^{2} z^{3} + 3 y^{2} z^{3} + x + y + z.
\end{eqnarray*}
is a 
LPP  in $\F_5[x,y,z]$ by Theorem \ref{composition}, and  has degree $9=3(5-2)$. 

\

 In fact the previous idea can be generalized for more general $q,n$.  We can prove the following theorem
 
 \begin{theorem}\label{bound_variables} Let $q=p$ prime and let $1\le n< p$  an integer such that 
 $\gcd(n,p-1)=1$. There exist an LPP  in $\F_p[\overline x]$ of  degree $n(p-2)$.
\end{theorem}
\begin{proof} Note that $f(x)=x^n$ and $g(x)=x^{p-2}$ are permutation polynomials in  $\F_p$, since $\gcd(p-1,n)=\gcd(p-1,p-2)=1$, see \cite{LN}.

\

Now by Theorems \ref{nearseparated} and  \ref{composition}, $h(\overline x)=(g(x_1)+\dots+g(x_n))^n$ is an LPP. So to prove the theorem it is enough to prove that the degree is $n(p-2)$. Note that this is equivalent to prove that there is a nonzero monomial of degree $n(p-2)$. Now let us call $y_i=x_i^{p-2}$ and $S_n=y_1+\dots+y_n$. Then
$$
h(\overline x)=S_n^n
$$
is a form of degree $n$, so all its monomials are of the form $Ay_1^{e_1}\dots y_n^{e_n}$, for $e_1+\dots+e_n=n$, so the only monomials divisible by $y_1\dots y_n$ are of the form $Ay_1\dots y_n$ for some $A\in\F_p$. Since
$$
S_n^n=(y_1+\dots+y_n)\dots(y_1+\dots+y_n),
$$
the monomial $y_1\dots y_n$ will appear only when selecting one distinct variable from each factor. Now, we have $n$ different factors to choose $y_1$, $n-1$ to choose $y_2$ and so on, until it remains  one factor to choose $y_n$, so in particular the monomial $y_1\dots y_n$ appears $n!$ times, which is non zero, since $p\nmid n!$. Hence $h(\overline x)$ has the non zero monomial $n!x_1^{p-2}\dots x_n^{p-2}$ of degree $n(p-2)$.
\end{proof}

For the case $p=3$, $n=2$, we know there is no LPP of sharp degree since we know that all the local permutation polynomials in $\F_3[x,y]$ are linear. For $q>3$ and  $n=2$ , following the same line of reasoning we get a new simpler proof of the result in  \cite{DHK}. For that we need the following lemma which gives the polynomial describing any permutation in $\F_q$ as the composition of  transpositions and  cycles of maximal length. The following result is partially cover in \cite{LN}.

\begin{lemma} \label{lem:permutation}The polynomial 
$$
f(x)=x+\sum_{k=0}^{q-2}x^k
$$
permutes $1$ and $0$, and leave fixed any other element in $\F_q$.
In general for any $a,b\in\F_q$
$$
f_{a,b;q}(x)=a+(b-a)\left(\frac{x-a}{b-a}+\sum_{k=0}^{q-2}\left(\frac{x-a}{b-a}\right)^k\right)
$$ 
is a permutation polynomial representing the transposition $(a b)$

\

On the other hand, if $\alpha$ is a primitive element in $\F_q^*$ then the polynomial 
$$
g_q(x)=(\alpha x-1)^{q-1}-x^{q-1}+\alpha x
$$ 
is a permutation polynomial representing a cycle of length $q$.
\end{lemma}
The proof is straightforward. 

\

Now we are in a position  to prove the following theorem.

\begin{theorem}\label{bound2_variables} For any $q>3$ a power of prime $q=p^s$ there exist an LPP  in $\F_q[x,y]$ of  degree $2(q-2)$.
\end{theorem}
\begin{proof} The case $\F_4$ is given by the example
$$
p(x,y)=u x^{2} y^{2} + \left(u + 1\right) x^{2} y + \left(u + 1\right) x y^{2}
+ x y + y^{2} + u x + 1,
$$
where $u^2+u+1=0$. So suppose $q\ge 5$, odd. Consider the polynomial in $\F_q[x,y]$ given by 
$$
P=x^{q-2}+y^{q-2}+\sum_{k=0}^{q-2}(x^{q-2}+y^{q-2})^k.
$$
It is an  LPP since it is the composition of an LPP and a permutation polynomial by Theorem \ref{nearseparated}  and Lemma \ref{lem:permutation}. Expanding it we have
\begin{eqnarray*}
P&=&x^{q-2}+y^{q-2}+\sum_{k=0}^{q-2}\sum_{j=0}^{k}\binom{k}{ j}x^{(k-j)(q-2)}y^{j(q-2)}\\
&=&x^{q-2}+y^{q-2}+\sum_{j=0}^{q-2}\left(\sum_{k=j}^{q-2}\binom{k}{ j}x^{(k-j)(q-2)}\right)y^{j(q-2)}.
\end{eqnarray*}
Now we have $j(q-2)\equiv q-2\pmod{q-1}$ only if $j\equiv 1\pmod {q-1}$. Selecting $k=2$ we have that $P$ has the term
$$
M=2x^{q-2}y^{q-2}\ne 0.
$$
For any other $j\ne 1$, $j(q-2)\not\equiv q-2 \pmod{q-1} $, meanwhile for $j=1$ and any other $k$ we have that $(k-j)(q-2)\not\equiv q-2\pmod {q-1}$  and hence  $M$ is the only monomial of degree $2(q-2)$.

\

Now suppose $q\ge 8$ a power of $2$, and let $q_2=\frac{q-2}2$. Consider
$$
P=x^{q-2}+y^{q_2}+\sum_{k=0}^{q-2}(x^{q-2}+y^{q_2})^k=x^{q-2}+y^{q_2}+\sum_{j=0}^{q-2}\left(\sum_{k=j}^{q-2}\binom{k}{ j}x^{(k-j)(q-2)}\right)y^{jq_2}.
$$
Again $jq_2\equiv q-2 \pmod{q-1}$ only if $j=2$ and on the other hand we have that $(k-2)(q-2)\equiv q-2 \pmod{q-1}$ only if $k=3$, so the only term in $P$ of degree $2(q-2)$ is $M=3x^{q-2}y^{q-2}\ne 0$.
  \end{proof}

\section{Bivariate local permutation Polynomials} \label{dosvariables}

  Local permutation polynomials in two variables $\F_q[x,y]$ correspond to  Latin squares of order $q$.
  This section  provides new families of local permutation polynomials in  $\F_q[x,y]$.

\subsection{Permutation polynomial tuples}

Let $\Sigma_q$ be the permutation group with $q$ elements and  $\F_q=\{c_0,\dots,c_{q-1}\}$ the field with $q=p^r$ elements. Given a  permutation polynomial $f \in \F_q[x,y]$, then for each $c_i\in \F_q$, $i=0,\ldots,q-1$,  we define the set
\begin{equation}\label{partition}
A_i=\{(a_{i,j},b_{i,j}),j=0,\dots,q-1\,:\, f(a_{i,j},b_{i,j})=c_i \}.
\end{equation}
Since $f$ is a permutation polynomial, it follows that $\{A_i,0\le i\ \le q-1 \}$ form a partition of $\F_q^2$ and that $|A_i|=q$.  Also, if we consider an LPP,  then 
we see that, for each $0\le i\le q-1$, there exist a permutation
$ \beta_i \in  \Sigma_q$ such that, 
\begin{equation}\label{eq:partition}
A_i=\{(c_j,\beta_i(c_j),j=0,\dots,q-1\,:\, f(c_j,\beta_i(c_j))=c_i \}, \quad i=0, \ldots,q-1,
\end{equation}
verifying $\beta_i(c_j)\ne \beta_k(c_j)$ for any $0\le i,j,k\le q-1$, and $i\ne k$,   since the sets $A_i$ are disjoint. In other words, $\beta_i^{-1}\beta_k$ has no fixed points. 

\

So, the above study  allows to describe  local permutation polynomials as $q$-tuples of permutations: 

\begin{lemma}\label{list_permutation} There is a bijective map between the set of local permutation polynomials  $f \in \F_q[x,y]$, and the set of  $q$-tuples $\underline{\beta}_f = (\beta_0,\dots,\beta_{q-1})$ such that $ \beta_i \in  \Sigma_q, \,( i =0,\ldots, q-1) $ and  for $i\ne j$, $\beta_i^{-1}\beta_j$ has no fixed points.
\end{lemma}
\begin{proof}
We have  already seen how to associate a $q-$tuple of permutation to a given LPP.  For the other direction, note that given a $q-$tuple $(\beta_0,\dots,\beta_{q-1})$ with $ \beta_i \in  \Sigma_q, \, i =0,\ldots, q-1$, and no fixed points as defined above,  we can construct the set $A_i$ as in equation \eqref{eq:partition}.
 Then Lagrange Interpolation  algorithm would return the polynomial,  completing the proof. 
  \end{proof}

We denote by  $\underline{\beta}_f = (\beta_0,\dots,\beta_{q-1})$ the $q$-tuple  associated to  the LPP  $f$ 
 as in Lemma \ref{list_permutation}.

\begin{remark} Note that the $q$-tuple can be similarly  defined acting on the first variable as
$$
A_i=\{(\beta_i(c_j),c_j),j=0,\dots,q-1\,:\, f(\beta_i(c_j),c_j)=c_i \}.
$$

\end{remark}

Let us illustrate the above result by an example:
  
\begin{example} Let $\F_9=\{c_0,c_1,\ldots, c_8\}=\{0, 1, 2, u, u+1, u+2, 2u, 2u+1, 2u+2\}$ such that $u^2+u+1=0$ and   $f= x^{5} + y^{5}$, then $(\beta_i, i=0,\ldots,8)$ are the product of four transpositions:

\begin{equation*} \begin{split}
\beta_0&=(1,2)(u,2u),(u+1,2u+2)(u+2,2u+1)   \\
\beta_1&=(0,1)(u,u+1)(u+2,2u)(2u+1,2u+2)\\
\beta_2&=(0,2)(u,2u+1),(u+1,u+2)(2u,2u+2) \\
\beta_3&=( 0, 2u)(1,2u+1)(2,u+1),(u+2)(2u+2)\\
\beta_4&=(0, u+1)(1,2u),(2,2u+1)(u,u+2)\\
\beta_5&=( 0, 2u+1)(1,u+1)(2,2u)(u,2u+2)\\
\beta_6&=(0,u)(1,2u+2)(2,u+2)(u+1,2u+1)\\
\beta_7&=( 0,u+2)(1,u)(2,2u+2)(u+1,2u) \\
\beta_8&=(0,2u+2)(1,u+2)(2,u)(2u,2u+1) \end{split}
\end{equation*}
\end{example}
The example has been created with SageMath, and it can also be used to verify that indeed,  for $i \not = j$ then $\beta_i^{-1}\beta_j$ has no fixed points.

\begin{remark}\label{permutation} Another interesting fact is that given an LPP
$f$, its associated partition  $A_i$ of $\F_q^2$, and any $\sigma \in \Sigma_q$ the sets $A_{\sigma(i)}$ for $i=0,\ldots,q-1$ form a new  partition of  $\,\F_{q}^2$, and consequently it provides a new LPP $g(f(x,y))$, where $g(z) \in \F_q[z]$ is the permutation polynomial associated to the permutation $\sigma$, see also Theorem \ref{composition}-(2).

\end{remark}

From Lemma \ref{list_permutation} we can translate the study of local permutation polynomials to the study of tuples
$(\beta_0,\dots,\beta_{q-1}) \in \Sigma_q^q$, such that  $\beta_i^{-1}\beta_j$ has no fixed point, for $i\ne j$. This suggests the following definition:

\begin{definition}  We say that  $(\beta_0,\dots,\beta_{q-1}) \in \Sigma_q^q$ is a {\it permutation polynomial tuple} if it  satisfies  that for $i\ne j$, $\beta_i^{-1}\beta_j$ has no fixed point.  
\end{definition}

From a {\it permutation polynomial tuple} we have $q!$  local permutation polynomials, just by permuting its elements, see Remark \ref{permutation}. In fact,  from one {\it permutation polynomial tuple}  we can construct many other local permutation polynomials as is shown in the next result:
 
\begin{proposition}\label{muchos_conjuntos} Let  $\Omega= (\beta_0,\dots,\beta_{q-1}) \in \Sigma_q^q$ be a  {\it permutation polynomial tuple} and let $\sigma, \delta  \in \Sigma_q$, then 
$ \sigma\Omega \delta =
(\sigma\beta_{0}\delta,\dots,\sigma\beta_{q-1}\delta) \in \Sigma_q^q$ is  also a  {\it permutation polynomial tuple}.
\end{proposition}
\begin{proof}
For $i\ne j$, if $c\in \F_q$ is a fixed point of $(\sigma\beta_i\delta)^{-1}(\sigma\beta_j\delta)$ then $\delta(c)$ is a fixed point of $\beta_i^{-1}\beta_j$, because  $$(\sigma\beta_i\delta)^{-1}(\sigma\beta_j\delta)= 
\delta^{-1} \beta_i^{-1}\beta_j\delta.$$
\end{proof}
The Proposition \ref{muchos_conjuntos} motivates the following concept:

\begin{definition}\label{equivalent}
Two permutation polynomial tuples $\Omega $ and $\Gamma$ are equivalent if there exit  $\sigma, \delta  \in \Sigma_n$ such that
 $\sigma\Omega \delta =  \Gamma$.  Similarly, we say that two local permutation polynomials $f$ and $g$ are equivalent if the corresponding permutation polynomial tuples  ${\underline \beta}_f$ and ${\underline \beta}_g$ are equivalent. 
\end{definition}

It is straightforward to check that the above is an equivalence relation defined in the set of local  permutation polynomials. Observe that every class has a representative containing the identity. If needed we will use this representative. 
We will see later that in $\F_2$ and $\F_3$  there is only  one equivalence relation class, and  two in $\F_4$.

\subsection{Permutation Group Polynomial }
A significant  {\it permutation polynomial tuple} is given by a permutation subgroup of $\Sigma_q$.

\begin{definition}\label{permutation_group} We say that  an LPP $f \in \F_q[x,y]$ is  a {\it permutation  group polynomial} if  $\{ \beta_0,\dots,\beta_{q-1} \}$ is subgroup of $\Sigma_q$ where
${\underline \beta}_f  = ( \beta_0,\dots,\beta_{q-1})$. We denote this subgroup by $G_{\underline{\beta}_f }$.

 \end{definition}
 Note that a subgroup of $\Sigma_q$ is  a {\it permutation polynomial tuple} if and only if it has no fixed points, i.e, it is a subgroup such that,  apart from the identity,  none of its elements has fixed points. 

Clearly,  if $C$ is a cycle of maximum  length $q$, then the cycle subgroup $<C>$ generated by $C$ is a group without fixed points.  Next, we describe another  family  of such subgroups.

We will denote $C$  to be a cycle of length  $|C|$. Sometimes, we will use a subindex in the cycle if we need to order cycles.

 \begin{lemma}\label{lem:elements} Let $q=p^r$, $G\subset \Sigma_q$  be a nontrivial subgroup  without fixed points, and $\alpha\in G$.  Then there is an $0 <e \le r$ such that for $t=p^e$ and   $k=p^{r-e}$ we have $\alpha=C_1\cdots C_k$ where $|C_i|=t$ for all $i=1,\dots k$.
 \end{lemma}
\begin{proof} Suppose $\alpha=C_1\cdots C_k$ is the representation of $\alpha$ as product of disjoint cycles, and suppose $|C_1|=t_1<t_2= |C_2|$. Then  $\alpha^{t_1}\in G$, is not the identity but fixes all the elements in $C_1$.  Hence, all the cycles have the same length, say, $t$.  Now by  Lagrange theorem there exits $0< e \le r$ such that $t=p^e$.  Finally,  we remark that $k \times t= p^r$, since each element of $\F_q$
should appear in that representation.
\end{proof} 
 In order to find subgroups without fixed points we will use the following technical result. Note that by Lemma \ref{lem:elements} the permutations will be products of cycles of the same lenght.
 
\begin{lemma}\label{alpha_beta} Let $q=p^r$, $1\le e\le r$, $l=p^{e}, t=\frac ql$, $\alpha=C_{0,\alpha}\cdots C_{t-1,\alpha}$, $\beta=C_{0,\beta}\cdots C_{l-1,\beta}$ such that for $0\le i\le t-1$
 $$
 C_{i,\alpha}=\{c_{j+il},j=0,\dots, l-1\}
 $$
 and  for $0\le j\le l-1$
 $$
 C_{j,\beta}=\{c_{j+il},i=0,\dots, t-1\}.
 $$
 Then for any $0\le a\le l-1$ and $0\le b\le t-1$,
   $\beta^b\alpha^a$ has no fixed points and $\alpha^a\beta^b=\beta^b\alpha^a$.
 \end{lemma}

\begin{proof}

  We write the elements of $\F_q$ as $c_{j+il}$ for some $0\le j\le l-1$ and $0\le i\le t-1$. Then
 $$
 \beta^b\alpha^a(c_{j+il})= \beta^b(c_{(j+a) \pmod l+il})=c_{(j+a)\pmod l+(i+b)\pmod t l}.
 $$
 This proves the first claim since  $(j+a)\not\equiv j\pmod l$ unless $a=0$, and in that case $(i+b)\pmod t \not\equiv i\pmod t$ unless $b=0$. Moreover  
  $$
 \alpha^a\beta^b(c_{j+il})= \alpha^a(c_{j+(i+b)\pmod t l})=c_{(j+a)\pmod l+(i+b)\pmod t l}, 
 $$
 which proves commutativity. 
 \end{proof}
 
 With the above  notations and definitions, let $C_\alpha$ be the matrix of  $t$ rows $C_{i,\alpha}, (i=0,\ldots, t-1)$ and  $l$ columns; let $C_\beta$ be the matrix of  $l$ rows $C_{j,\beta}, (j=0,\ldots, l-1)$ and  $t$ columns;
 
  $$ C_\alpha= \left(\begin{array}{cccc} c_{0} &  c_{1} &  \ldots &  c_{l-1} \\c_{l} &  c_{l+1}&  \ldots &  c_{2l-1}\\\cdots &  \cdots &  \cdots &  \cdots\\c_{(t-1)l} &  c_{(t-1)l+1}&  \ldots &  c_{q-1}\\\end{array}\right), \quad 
  C_\beta= \left(\begin{array}{cccc} c_{0} &  c_{l} &  \ldots &  c_{(t-1)l} \\c_{1} &  c_{l+1}&  \ldots &  c_{(t-1)l+1}\\\cdots &  \cdots &  \cdots &  \cdots\\c_{l-1} &  c_{2l-1}&  \ldots &  c_{q-1}\\\end{array}\right)$$

 Notice that  $C_\alpha$ is the transpose matrix of $C_\beta$.

 \begin{corollary}\label{klenians} Let $\alpha,\beta$ be as in the previous Lemma \ref{alpha_beta}. Then  the set defined by $G=\{\alpha^i\beta^j:0\le i\le l-1,0\le j\le t-1\}$ is a subgroup of $\F_q$ without fixed points and order $|G|=q$.
 \end{corollary}
 \begin{proof}
 
We have already seen that it is a group without fixed points, so the only thing to see is that $|G|=q$, which follows since clearly $\alpha^a\beta^b$ are all distinct when $0\le a\le l-1,0\le b\le t-1$.

\end{proof}

The previous study suggests the following definition:

 \begin{definition} We will call an $e$-Klenian subgroup to any group of the form given in the Corollary \ref{klenians}. Also
 we say that a polynomial $f \in \F_q[x,y]$ is an $e$-Klenian polynomial if  $f$ is a  permutation 
 group polynomial and the associated group  $G_{\underline{\beta}_f }$ is  an $e-$Klenian subgroup. 

\end{definition}

Of course, there are groups without fixed points that are not $e$-Klenian's ones, and consequently there are permutation group polynomials that are not $e$-Klenian's polynomial. However, in practice it is hard to distinguish what type of polynomial is  only by looking at their formula. For example in the field of $8$ elements $\F_8=\{0, u, u^{2}, u + 1, u^{2} + u, u^{2} + u + 1, u^{2} + 1, 1\}$, where 
$u^3+u+1=0$, we have the $0$-Klenian polynomial given by the tuples $\{\beta^i:i=0,\dots 7\}$ for $\beta=(0, 1, u^{2} + 1, u, u^{2} + u, u^{2}, u^{2} + u + 1, u + 1)$ with 45 monomials and degree $11$, or the $1$-Klenian polynomial with tuple generated by the permutations $\alpha=(0,u)(u^2,u^3)(u^4,u^5)(u^6,u^7)$, $\beta =(0,u^2,u^4,u^6)(u,u^3,u^5,u^7)$, with $44$ monomials and degree $12$.

\

On the other hand  we have the following group polynomials, not $e$-Klenians. The first is associated with the tuple given by the non-abelian group of order $8$ $H_1 = <\alpha, \beta> $  generated by 
$$\alpha=(0,u)(u^2,u^3)(u^4,u^5)(u^6,u^7), \quad \beta =(0,u^3)(u^4,u)(u^2,u^7)(u^5,u^7),$$
each of them being the product of four  disjoint cycles of length $2$.
It is straightforward to check 
that the subgroup  $H_1$ has no fixed points. In this case the  local permutation group  polynomial associated to $H$ has degree $12$ and $46$ monomials.
\

 Now,  we consider another non-abelian group of order $8$ without fixed points $H_2=<\alpha,\beta> $ generated by permutations which are  the product of two  disjoint cycles of length $4$:
$$\alpha=(0,u,u^2,u^3)(u^4,u^5,u^6,u^7), \quad \beta =(0,u^4,u^2,u^6)(u,u^5,u^3,u^7)$$
 In this case the local permutation group  polynomial associated with $H_2$  has degree $10$ and  
 $42$ monomials. 

\

Not only distinguish $e$-klenians polynomials, but only count them all is non trivial.  We have not seen in the literature  significant results on this finite group problem.

\

 On the other hand, this problem has a straightforward solution   when we restrict to $e=0$.  Indeed the number of cycles of maximal lenght in $\Sigma_q$ is  $(q-1)!$, and a subgroup generated by a  cycle of length $q$  contains exactly $\varphi(q)$ generators, the prime powers of the cycle, so  the number of $0$-Klenian groups of $\Sigma_q$  is
$\frac{(q-1)!}{\varphi(q)}$ . Now, for each group, we need to order its elements to get the partitions associated to the polynomial, so the total number of $0$ klenian polynomials  in $\F_q$ for $q=p^r$ is 
\begin{equation} \label{number_cycle2}
\frac{q!(q-1)!}{\varphi(q)}= \frac {p^r!(p^r-1)!}{p^{r-1}(p-1)} .
 \end{equation}

Let us not that $e=0$ is the only case appearing  when we restrict to prime fields  $\F_p$, since  any permutation group polynomial of  $\F_p[x,y]$ should be a cycle subgroup of order $p$. In fact,  since any two cycle subgroups of $\Sigma_q$ are conjugated,  all e-Klenian polynomials in $\F_p[x,y]$ are equivalent. We can generalize a bit this result to  the following:

\begin{lemma}\label{eklenian_equiv}  Let $h \in \F_q[x,y]$ an LPP defined by  $\underline{\mu}_h = (\mu_0,\dots,\mu_{q-1})\in\Sigma_q$. Then, $h$ is equivalent to an $e-$Klenian polynomial if and only if  for any $1\le n\le q-1$, $\mu_n=\mu_{i,j}=\mu_{0} \alpha^i\beta^j$, where $n=i+jl$ for some $0\le i\le l-1,0\le j\le t-1$  and $G=\{\alpha^i\beta^j:0\le i\le l-1,0\le j\le t-1 \}$ is an $e$-Klenian group for $l=p^e$ and $t=p^{r-e}$.
 \end{lemma}
 
 \begin{proof} Suppose $f$ is equivalent to an $e$-Klenian group $G=\{\alpha^i\beta^j:0\le i\le l-1,0\le j\le t-1 \}$ as in Corollary \ref{klenians}. Then for some permutations $\sigma,\gamma$ we have
 $$
 \sigma \alpha^i\beta^j\gamma=\mu_{i,j},
 $$
 hence 
 \begin{eqnarray*}
 \mu_{i,j}^{-1}\mu_{I,J}&=&\gamma^{-1}\alpha^{I-i}\beta^{J-j}\gamma=(\gamma^{-1}\alpha^{I-i}\gamma)(\gamma^{-1}\beta^{J-j}\gamma)\\
 &=&(\gamma^{-1}\alpha\gamma)^{I-i}(\gamma^{-1}\beta\gamma)^{J-j}=\hat\alpha^{I-i}\hat\beta^{J-j}
 \end{eqnarray*}
 where  $\hat G=\{\hat \alpha^i\hat \beta^j:0\le i\le l-1,0\le j\le t-1 \}$ is also an $e$-Klenian group for $e=p^l$
 since conjugates of cycles are cycles of the same  length. Now, note that in particular
 $$
 \mu_{i,j}^{-1}\mu_{0}=\hat \alpha^{-i}\hat\beta^{-j}
 $$
 so
 $$
 \mu_{i,j}=\mu_{0}\hat \alpha^{i}\hat \beta^j,
 $$
 as wanted.
 
\end{proof}

\begin{corollary}\label{equivalentklein} There are exactly $(q-1)!N$ local permutation polynomials equivalent to $e$-Klenian polynomials over $\F_q$, where $N$ is the number of $e-$Klenian polynomials.
In particular if $q=p$ is prime, we have exactly $p!(p-1)!(p-2)!$ local permutations polynomials equivalent to a $0$-Klenian polynomial.\end{corollary}

\begin{proof} Every polynomial equivalent to an $e$-Klenian polynomial is of the form $\mu G$ where $\mu\in\Sigma_q$ and $G$ is the $q$-tuple of an $e$-Klenian polynomial. Now, the only way of getting two equal polynomials would be if we have 
$\mu_1G_1=\mu_2G_2$ and hence, $G_1=\mu_1^{-1}\mu_2G_2$ but then since  $G_2$ contains the identity, $\mu_1^{-1}\mu_2\in G_1$ and, since $G_1$ is a group its inverse is also in $G_1$ so we get that $G_2=\mu_2^{-1}\mu_1G_1=G_1$ so for each $G$, $\mu G$ gives new polynomials unless $\mu\in G$. Since we have $q!$ permutations, we get $(q-1)!N$ equivalent polynomials, as wanted.  The  proof of the second assert follows from 
Equation \ref{number_cycle2}.
\end{proof}

\subsection{Local permutation polynomials in $\F_2$, $\F_3$,  $\F_4$ and $\F_5$}
In this subsection we show that all local permutation polynomials over the $\F_2$, $\F_3$ and $\F_4$ are described by $e$-Klenian polynomials.

\subsubsection{ The finite field $\F_2$}
In this case the degree is $q-1=1$, and hence the only local permutation polynomials over $\F_2= \{0,1\} $ are $x+y$ and $x+y+1$, which correspond to
the only {\it permutation polynomial set} $\Omega=\{(I_d,\beta) \} \subset \Sigma_2$, where $\beta=(0,1)$ is the only cycle
of length $2$. The two polynomials appear from the two permutations of  the two elements of $\Omega$. 

\subsubsection{The finite field $\F_3$} It is known that the number of local permutation polynomials over the field $\F_3= \{0,1,2\}$ is 12, see  \cite{M1} and, by Corollary \ref{equivalentklein} we see that they are alll equivalento to $e$-Klenian polynomials. In fact, we have one $0$-Klenian subgroup 
 of $\Sigma_3$  generated by the cycle $\beta=(0,1,2)$, giving six $0$-Klenian polynomials by Equation \ref{number_cycle2}, and another $6$ equivalent to them.

\subsubsection{The finite field $\F_4$}
It is known that the number of Latin squares of order $4$ are 576, so we have the same number  of local permutation polynomials of $\F_4$. We will use the following description $\F_4=\{0,u, u^2, u^3\}= \{0,u, u+1, 1\}$ such that $u^2+u+1=0$.  In total there are $4$ $e$-Klenian subgroup . With $e=0$,  there are three cycle groups of order $4$ generated by $\beta_i$, for $i=1,2,3$, 
\begin{enumerate}
\item[1.] $K_1=<\beta_1=(0,u,u^2,u^3)>$ 
\item[2.] $K_2=<\beta_2=(0,u^2,u,u^3)>$ 
\item[3.] $K_3=<\beta_3=(0,u^2,u^3,u)>$ 
\end{enumerate}
and by Equation \ref{number_cycle2}, these give $72$ $e$-Klenian polynomials, producing $432$ polynomials equivalent to them by Corollary   \ref{equivalentklein}. Finally for   $e=1$ we have a  group  generated by $\alpha=(0,u)(u^2,u^3)$ and 
$\beta=(0,u^2)(u,u^3)$
\begin{enumerate}
\item[4.] $K_4=\{ \alpha, \beta, \alpha\beta, \alpha^2\}$.
\end{enumerate}
giving $24$ $1$-Klenian polynomials, and  again by Corollay  \ref{equivalentklein} $144$ local permutation polynomials equivalent to them, giving a total of $576$.

\subsubsection{The finite field $\F_5$}

These constructions do not complete the list in  other fields of the cardinality bigger than $4$. In $\F_5$, the number of e-Klenian subgroups is $6$, giving  $720$ $0$- Klenian polynomials by Equation   \ref{number_cycle2}, and producing  $17280$ local permutation polynomials equivalent to them by Corollary \ref{equivalentklein}.

\

On the other hand, it is known that the number of Latin squares of order $5$ are $161280$. The next  example
shows an LPP  of degree $6$ that it is not one obtained  by e-Klenian polynomials.

\begin{example}  We will construct a polynomial over $\F_5$ non equivalent to a $0-$Klenian polynomial.  We need a $5$-tuple 
$\{\beta_0,\dots,\beta_4\}$ so that $\beta_{j}^{-1}\beta_i$ has no fixed points for any $0\le i<j\le 4$. So, we first select 
$\beta_0\in\Sigma_5$ at random. Now, we will need to find $\alpha_1,\dots,\alpha_4$ with no fixed points, and consider
\begin{equation}\label{eq:tuple}
\beta_i^{-1}\beta_0=\alpha_i.
\end{equation}
In order to find an appropriate tuple for an LPP we need  $\alpha_i$ to verify another condition, namely
$
\alpha_i\alpha_j^{-1} 
$
to be with no fixed points. Observe that this is similar to the condition on the $\beta$'s but now $1\le i<j\le 4$. So, we continue this process and,  next, we select at random $\alpha_1\in\Sigma_5$ any permutation with no fixed points, and try to find $\gamma_2,\gamma_3,\gamma_4$ without fixed points so that
$$
\gamma_4\gamma_2^{-1},\quad \gamma_3\gamma_2^{-1},\quad \gamma_4\gamma_3^{-1},\quad \alpha_1\gamma_2^{-1},\quad \alpha_1\gamma_3^{-1},\quad \alpha_1\gamma_4^{-1}
$$
has no fixed points. This will give
$$
\alpha_i=\alpha_2\gamma_i^{-1},\quad \text{ for } i=2,3,4.
$$
and then the needed tuple by (\ref{eq:tuple}). We start with  any permutation $\beta_0$, for example $\beta_0=(0,1)$, and now since the roles of $\beta$'s and $\gamma$'s is similar we take  $\alpha_1=(0,1)(2,3,4)$, $\gamma_2=(0,2)(1,4,3)$,$\gamma_3=(0,3)(1,2,4)$, $\gamma_4=(0,4)(1, 3,2)$. This selection gives us the LPP
$$
f(x,y)=2x^3y^3 + 4x^2y^3 + 2x^3y + x^2y^2 + 4x^3 + 2x^2y + 4xy^2 + 4y^3 + 2xy + 1
$$
Note that we get 
$$
\begin{matrix}\alpha_1&=&(0,1)(2,3,4)\\\alpha_2&=&(0,3,2,1,4)\\\alpha_3&=&(0,4,3,1,2)\\\alpha_4&=&(0,2,4,1,3)\end{matrix}
$$
which are not successive powers of a cycle of maximal length, so it can not be equivalent to a  $0-$Klenian polynomial because  Lemma \ref{eklenian_equiv}.
\end{example}

\begin{example} It is well known that there are two isotopy class of latin squares of size $5$, one of them is equivalent to 0-Klenian polynomial and the other one is not, see \href{http://users.cecs.anu.edu.au/~bdm/data/latin.html}{http://users.cecs.anu.edu.au/bdm/data/latin.html}, so we consider the LPP in $\F_5[x,y]$  in the non equivalent class given by 

\medskip

 $f= 2 x^{3} y^{3} + 2 x^{3} y^{2} + 3 x^{2} y^{3} + 2 x^{3} y + 2 x y^{3} +x^{2} y + 2 x y^{2} + 2 x y + x + y$.

\medskip
 
\noindent  The   associated permutation polynomial tuple is given by $\underline{\beta}_f =(\beta_0,\beta_1,\beta_2,\beta_3, \beta_4)$ where
\begin{eqnarray*}
\beta_0 &=& (1,2,4,3) \\
\beta_1& = &(0,1)(2,3,4) \\
\beta_2 &=&(0,2)(1,4) \\
\beta_3 &=& (0,3) \\
\beta_4 &=& (0,4),(1,3,2)
\end{eqnarray*}
Again, by  Lemma \ref{eklenian_equiv} the LPP $f$  can not be equivalent to a  $0-$Klenian polynomial.
\end{example}

\section{Orthogonal system of  polynomials and  Mutually Orthogonal Latin Squares}\label{mols}

Let us recall  the Latin square's definition. In this paper we only consider Latin squares of order a prime power.

\begin{definition} A  latin square  of order $q$ is a $q\times q$ matrix $L$ with entries from $\F_q$ such that each element of $\F_q$  occurs exactly once in every row and every column of $L$.
\end{definition}

See \cite{LM} for several properties and applications of Latin squares. Further relevance of the use of local permutation polynomials for the study of Latin squares or cubes are described in \cite{M1} and \cite{M2}.

\

By indexing the cells of $L$ by  $\F_q^2$, we have the following known result:
\begin{lemma}\label{latin}  There is a bijective map between Latin squares of order $q$  and local permutation polynomials of $\F_q[x,y]$. 
\end{lemma}
\begin{proof}
 Indeed,  given a Latin square $L$ over $\F_q$ with entries $a_{i,j} \in \F_q$, we consider the Lagrange interpolation polynomial with values $f(c_i,c_j)=a_{i,j}$. Note that, dividing by $x^q-x$ and $y^q-y$ we can assume  $\deg_{x}(f) < q$ and $\deg_{y}(f) < q$. The converse  is clear.
  \end{proof}

We now introduce the orthogonality property of Latin squares:
\begin{definition}
Two Latin squares $L_1$ and $L_2$ of order $q$ are called orthogonal Latin squares if
$$(L_1(i_1, j_1), L_2(i_1, j_1))\not = (L_1(i_2, j_2), L_2(i_2, j_2))$$

for all distinct pairs of coordinates $(i_1, j_1), (i_2, j_2) \in \N^2$.

\end{definition}
Equivalently,  two Latin squares of the same size (order) are said to be orthogonal if, when superimposed, each position has
a different pair of ordered entries. In terms of polynomials, the following classical definition appears in \cite{N1}:

\begin{definition}\label{orto} 
A system of $m$ polynomials $f_1,\ldots, f_m \in \F_q[x_1,\dots, x_n]$, for $ \, 1 \leq m \leq n$,  is said orthogonal  in $\F_q$ if the system of equations
$$f_1(\overline x) = a_1,\ldots, f_m(\overline x) = a_m $$ 
has $q^{n-m}$  solutions in $\F_q^n$ for each $(a_1,\ldots,a_m)\in \F_q^m$.
\end{definition}

 In the special case $m=1$,  a permutation polynomial alone forms an orthogonal system.  On the other hand, if $m=n$  this means that the orthogonal system  $f_1,\ldots,f_n$ induces a permutation of $\F_q^n$. 
These permutations are completely classified in  \cite{N1} for the special case when the orthogonal system contains polynomials of degree at most two. See also \cite{LN} for further properties and results about those interesting systems.

 An immediate consequence of Definition \ref{orto} and Lemma \ref{latin} is the following:

\begin{corollary}
Two latin squares $L_1 $ and $L_2$ are orthogonal if and only the associated polynomials  is an orthogonal system. 
\end{corollary}

The main goal in this part of the paper is constructing families of orthogonal latin squares. So, this bring the  next definition:

\begin{definition} Given a permutation polynomial $f \in \F_q [x,y]$ we say that $g$ is a companion of $f$ if 
$(f,g):\F_q^2\to \F_q^2$ defines a permutation, that is,  $f,g$ is an orthogonal system.
\end{definition}
 Obviously  any companion must be a permutation polynomial.
The following result count the number of companions:

\begin{theorem} A permutation polynomial $f$ has exactly $q!^q$ companions. 
\end{theorem}
\begin{proof}
 We consider  the partition of $\F_q^2$ given in Equation (\ref{partition})
 $$
A_i=\{(a_{i,j},b_{i,j}),j=0,\dots,q-1\,:\, f(a_{i,j},b_{i,j})=c_i \} ,\quad i=0,\ldots,q-1.
$$ 
Now, consider a $q$-tuple, $\{\sigma_1,\dots, \sigma_q\}\subset \Sigma_q$, and define the polynomial $g$ such that, $g(a_{i,j},b_{i,j})=\sigma_i(c_j),\quad j=0,\dots, q-1$. Now, every pair $(c_i,c_k)\in\F_q^2$ can be determined uniquely as $(c_i,\sigma_i(c_j))$ and, hence, the equation $(f,g)=(c_i,c_k)$ has exactly one solution for each pair $(i,k)\in[1,\dots,q]^2$.  Hence, each selection of $q$-tuple gives a different $g$ so, in particular we have $q!$ ways of choosing each $\sigma_i$ and in total there are $q!^q$ companions.

\
On the other hand, if $g$ is a companion, $g(A_i)=\F_q$ and clearly there is a bijection $h_i:\F_q\to A_i$, so there is a $q$-tuple of permutations $\sigma_i=g\circ h_i$ associated to $g$.
\end{proof}
 
The problem is more interesting when we consider local permutation polynomials, that is, Latin squares.

\begin{question}\label{question:q} Is it true that any LPP has a companion which is also an LPP?
\end{question}

The answer in general is no. For example,  for $q=2$,  the only local permutation polynomials are $x+y$ and $x+y+1$, and obviously is no an orthogonal system of polynomials. For $q=4$ we find after some computations with SageMath, that only $144$ of the total of $576$ local permutation polynomials that exist in $\F_4$ have LPP companions, and each of them has exactly $48$  companions.

\

In general we have several ways to find orthogonal systems. First, if we restrict to the linear case  we have the following theorem
\begin{theorem} For $q\ge 3$, every linear LPP has  companions which is also a linear LPP.
\end{theorem}
\begin{proof} Let $f(x,y)=ax+by+c$ be an LPP. Observe that any linear permutation of this form with $ab\ne 0$ is indeed an LPP, trivially. Now consider $g=ux+vy+w$ so that $av-bu\ne 0$. Then $(f,g)$ are companions since any linear system with non zero determinant  has a unique solution. Observe that, in general, permutation polynomials have many companions. We can take for example $v=(c+1) b$, $u=c a$ for any $c\in\F_q$, $c\ne0,-1$.  The same example serves to see that different polynomials can share the same companion.
\end{proof}
\

Also, given an orthogonal system, we construct new ones with the following simple result.

\begin{proposition}\label{compa} If $f(x,y), g(x,y)$ is an orthogonal system,  then the polynomials $af(x,y)+bg(x,y), cf(x,y)+dg(x,y)$ form also an orthogonal system for $a,b,c,d \in \F_q$ such that $ad-bc \not= 0$.
\end{proposition}
\begin{proof}
For any pair  $(c_i,c_j)\in\F_q^2$,  the system of equations:
\begin{eqnarray*}
af(x,y))+bg(x,y)&=&c_i\\
cf(x,y)+dg(x,y)&=&c_j
\end{eqnarray*}
has a unique solution, just   inverting the matrix $A= \left(\begin{array}{cc}a & b \\c & b\end{array}\right)$
\end{proof}

Another family of orthogonal system is provided by separated variable polynomials:

\begin{proposition} \label{separated} Let $f(z), g(z), h_1(z), h_2(z) $ be permutation polynomials in $\F_q[z]$, then $f(ah_1(x)+bh_2(y)), g(ch_1(x)+dh_2(y))$ is a orthogonal system for  $a,b,c,d \in \F_q$ such that $ad-bc \not= 0$.
\end{proposition}
\begin{proof}
For any pair  $(c_i,c_j)\in\F_q^2$,  the system of equations:
\begin{eqnarray*}
ah_1(x)+bh_2(y)&=&f^{-1}(c_i)\\
ch_2(x)+dh_2(y)&=&g^{-1}(c_j)
\end{eqnarray*} 
has a unique solution, just   inverting the matrix $A= \left(\begin{array}{cc}a & b \\c & b\end{array}\right)$,  since  $h_1(x)$,$h_2(y)$  are permutation polynomials. 

\end{proof}

Now, we are introducing an important concept related to  Latin squares:
\begin{definition} 
 A set of Latin squares, all of the same order, such that all pairs of which are orthogonal is called a set of Mutually Orthogonal Latin Squares (MOLS).   A set of $t>1$ MOLS of order n is called a complete set if $t= n-1$.
\end{definition}

 The following are very well know results, see \cite{LM}. 
\begin{theorem} Let $N(n)$  be the size of the largest collection of MOLS of order $n$, then we have
\begin{itemize}
\item $N(n) \leq n-1$.
\item If $q$ is a power of prime, then $N(q) = q-1$
\end{itemize}
\end{theorem}
As a trivial consequence of Proposition \ref{compa} and Proposition \ref{separated} we have two different complete set of MOLS:
 
 \begin{theorem} With the above notations and definitions:
 \begin{itemize}
 \item If $f(x,y)$ is a local permutation polynomial and $g(x,y)$ is any  LPP companion of $f(x,y)$  then the set $\{f(x,y) +ag(x,y), a \in \F_q^*\}$  is a complete set of MOLS. 
 \item If $f(x), h(y)$ are permutation polynomials, then the set $\{f(x) +ah(y), a \in \F_q^*\}$ is complete set of MOLS.
 \end{itemize}
 \end{theorem}
   
  The main result of this section  is the following:

  \begin{theorem} \label{importante} Let $2\nmid q$. Every $f$ $e$-Klenian polynomial  has a  companion which is  an LPP \end{theorem}
\begin{proof}
  Let $f(x,y)$ be an  $e$-Klenian polynomial and for each $m=0,\dots, q-1$ written as $m=a+bl$, for $0\le a\le l-1$, $0\le b\le t-1$, consider  the set 
  $A_m=\{(c_j, \alpha^a\beta^b(c_{j})),j=0,\dots,q-1\}$.
 We will see that $g$ to be defined by $B_{m}=\{(c_k, \alpha^{a+i}\beta^{b+j}(c_{k})),k=i+jl,\, 0\le i\le l-1, 0\le j\le t-1\}$ for $m=0,\dots, q-1$, is an LPP which is companion of $f$. 
 
 \
 
 First we see that $g$ is LPP. We start by proving that for any $c_k,c_m\in\F_q$, there exist a $y\in\F_q$ such that $g(c_k,y)=c_m$. As in the definitions before, let $k=u+lv$, with   $0\le u\le l-1$, $0\le v\le t-1$, and  $m=a+bl$, $0\le a\le l-1$, $0\le b\le t-1$. Then  $y=\alpha^{a+u}\beta^{b+v}(c_{k})=c_{(a+2u)\pmod l+(b+2v)\pmod t l}$, verifies the condition, i.e.  $(c_k,y)\in B_m$, by definition.
 
 \

 Now, we want to prove that $g$ is also a permutation polynomial  in the first variable, in other words that given $c_k,c_m\in\F_q$ as before, there exist $x$ such that $g(x, c_{k})=c_{m}$. In particular, we need to find $i,j$ such that $c_{k}=\alpha^{a+i}\beta^{b+j}(c_{i+jl})$. Indeed, in this case $x=c_{i+jl}$ is the solution needed, since by definition $(x,c_k)\in B_m$. But this is only possible 
 if  $a+2i\equiv u\pmod l$ and $b+2j\equiv v\pmod t$, or $i=\frac{u-a}{2}\pmod l$, $j=\frac{v-b}{2}\pmod t$.
 
 \
 
 Finally we need to see that $(f,g)$  is an orthogonal system or, in other words, that for any $c_m,c_k\in\F_q$ as before 
 \begin{eqnarray*}
 f(x,y) &= c_m\\
 g(x,y)  &= c_k
 \end{eqnarray*}
 has exactly one solution. Now, we take the set $A_{m}=\{(c_{i+jl}, \alpha^a\beta^b(c_{i+jl})),0\le i\le l-1,0\le j\le t-1\}$, and we need to check whether, for some $0\le i\le l-1,0\le j\le t-1$:
 $$(c_{i+jl}, \alpha^a\beta^b(c_{i+jl}))=(c_{i+jl}, \alpha^{u+i}\beta^{v+j}(c_{i+jl})).
 $$ 
 But then $i+a \pmod l=u+2i \pmod l$ and $b+j\pmod t=v+2j\pmod t$, or $i=a-u \pmod l$,$j=b-v\pmod t$ is the unique solution, so indeed $(f,g)$ are companions,
(observe that if $k=l$ then  $0$ is simply $q$).

\end{proof}

 Note that, even though we have proved $g$ to be a  LPP, we have not given explicitly  the associated permutation polynomial tuple in Lemma \ref{list_permutation}.
The following examples illustrate the above result:

\begin{example}

We consider two cases, a prime finite field $\F_7$ and  the finite field $\F_9$:

\textbullet \, Let $\beta =(2, 0,1, 3, 5, 6, 4)$ the cycle of length $7$, so the corresponding $e-$Klenian polynomial $f$ is:
$$  x^{5} -  y^{5} -  x^{4} + y^{4} + 3 x^{3} + 4 y^{3} + 2 x^{2} + 5 y^{2} + x -  y + 6$$ and the LPP produced in the above Theorem is 

$$ 2 x^{5} -  y^{5} + 5 x^{4} + y^{4} -  x^{3} + 4 y^{3} + 4 x^{2} + 5
y^{2} + 2 x -  y + 4$$
Then,  $f,g$ is an orthogonal set. 

$$          	\left(\begin{array}{rrrrrrr}
6 & 0 & 5 & 1 & 4 & 2 & 3 \\
5 & 6 & 4 & 0 & 3 & 1 & 2 \\
0 & 1 & 6 & 2 & 5 & 3 & 4 \\
4 & 5 & 3 & 6 & 2 & 0 & 1 \\
1 & 2 & 0 & 3 & 6 & 4 & 5 \\
3 & 4 & 2 & 5 & 1 & 6 & 0 \\
2 & 3 & 1 & 4 & 0 & 5 & 6
\end{array}\right), \quad 
\left(\begin{array}{rrrrrrr}
4 & 5 & 3 & 6 & 2 & 0 & 1 \\
2 & 3 & 1 & 4 & 0 & 5 & 6 \\
6 & 0 & 5 & 1 & 4 & 2 & 3 \\
0 & 1 & 6 & 2 & 5 & 3 & 4 \\
1 & 2 & 0 & 3 & 6 & 4 & 5 \\
5 & 6 & 4 & 0 & 3 & 1 & 2 \\
3 & 4 & 2 & 5 & 1 & 6 & 0
\end{array}\right) $$

\textbullet \, $\F_9=\{0, u, u + 1, 2 u + 1, 2, 2 u, 2 u + 2, u + 2, 1\}=
\{0, u, u^{2}, u^3, \ldots, u^{8} \} $ such that $u^2+2u+2=0$. Now,  let
 $\beta =(0, 2u+1, u+2, u, u+1, 2u+2, 1, 2, 2u)$ the cycle of length $9$, so the associated  $e-$Klenian polynomial  $f$ has $58$ non zero monomials and degree $14$.
 The permutation polynomial $g$ provided by the above result  has $57$ nonzero monomials and also degree $14$.
We have that $f,g$ is an orthogonal system.

$$\hskip-10pt        \left(\begin{array}{rrrrrrrrr}
1 & u + 1 & 2 u + 1 & 0 & 2 u + 2 & u + 2 & 2 & u & 2 u \\
2 u & 1 & 0 & 2 u + 2 & 2 u + 1 & 2 & u & u + 2 & u + 1 \\
2 & u + 2 & 1 & 2 u & u + 1 & 2 u + 1 & 0 & 2 u + 2 & u \\
u + 2 & u & u + 1 & 1 & 2 u & 2 u + 2 & 2 u + 1 & 0 & 2 \\
u & 2 & 2 u & u + 1 & 1 & 0 & 2 u + 2 & 2 u + 1 & u + 2 \\
0 & 2 u + 1 & 2 & u & u + 2 & 1 & 2 u & u + 1 & 2 u + 2 \\
2 u + 1 & 2 u + 2 & u + 2 & 2 & u & u + 1 & 1 & 2 u & 0 \\
2 u + 2 & 0 & u & u + 2 & 2 & 2 u & u + 1 & 1 & 2 u + 1 \\
u + 1 & 2 u & 2 u + 2 & 2 u + 1 & 0 & u & u + 2 & 2 & 1
\end{array}\right),$$
$$
\hskip-10pt\left(\begin{array}{rrrrrrrrr}
u + 2 & u & u + 1 & 1 & 2 u & 2 u + 2 & 2 u + 1 & 0 & 2 \\
u & 2 & 2 u & u + 1 & 1 & 0 & 2 u + 2 & 2 u + 1 & u + 2 \\
1 & u + 1 & 2 u + 1 & 0 & 2 u + 2 & u + 2 & 2 & u & 2 u \\
2 u & 1 & 0 & 2 u + 2 & 2 u + 1 & 2 & u & u + 2 & u + 1 \\
u + 1 & 2 u & 2 u + 2 & 2 u + 1 & 0 & u & u + 2 & 2 & 1 \\
0 & 2 u + 1 & 2 & u & u + 2 & 1 & 2 u & u + 1 & 2 u + 2 \\
2 u + 2 & 0 & u & u + 2 & 2 & 2 u & u + 1 & 1 & 2 u + 1 \\
2 u + 1 & 2 u + 2 & u + 2 & 2 & u & u + 1 & 1 & 2 u & 0 \\
2 & u + 2 & 1 & 2 u & u + 1 & 2 u + 1 & 0 & 2 u + 2 & u
\end{array}\right) $$

\end{example}

On the other hand, we can not omit the condition  $2\nmid q$ in Theorem \ref{importante}. For instance in the field $\F_4 =\{0,1,u, u+1\}$ with $u^2+u+1=0$, the  polynomial  
$$f=u x^{2} y^{2} + \left(u + 1\right) x^{2} y + \left(u + 1\right) x y^{2}
+ x y + y^{2} + u x + 1 \in \F_4[x,y]$$
is an e-Klenian's one defined by the $(\beta, \beta^2, \beta^3, \beta^4))$ where $\beta$ is the 4-cycle  $\beta=(0,1,u,u+1)$ and has not any companion LPP as we have checked by SageMath, see also the comments after Question \ref{question:q}.

\section {Conclusions and open problems}\label{conclusiones}

Contrary to the many papers and results on permutation polynomials in one variable, there are few  for local permutation polynomials in several variables. 

 We have  presented  some new  ideas, concepts  and results in the study of these kind of polynomials. In particular, in Theorem \ref{bound2_variables}  we have  elegantly shortened the proof  of \cite{DHK} and 
  generalised   obtaining   in Theorem \ref{bound-degree} the  bound $n(q-2)$   for the degree of local permutation  polynomial $f \in \F_q[\overline x]$ and 
 a sharp bound $n(p-2)$ Theorem \ref{bound_variables} for   polynomials defined in a prime finite field $\F_p$ if  $\gcd (n, p-1)= 1$. It should be interesting to investigate for what prime finite  fields  $\F_p$ the condition  $\gcd(n,p-1)=1$ could be avoided,  or more generally, for arbitrary finite fields $\F_q$.  We think that better results are expected. 
 
 \

 We have translated the study of local permutation polynomials to the study of permutation polynomial sets, (see Lemma \ref{list_permutation} and Proposition \ref{muchos_conjuntos}).  We believe this relationship opens a wide line of research in order to  investigate very deeply this relationship.
 
\

Clearly, a significant family of local permutation polynomials are the so called local permutation group polynomial, see Definition \ref{permutation_group}. We have describe here a small subfamily, the so called $e-$Klenian polynomials.
Giving  others rigorous  subclass of such permutation group polynomial  is a challenging open  problem as well.

\

Among other things, this will provide lower bounds  in the number of local permutation polynomials and, hence,  latin squares. Recall that the precise number of latin squares is an open problem with a lot of interest in the mathematical community in the area.

\

We have created several systems of orthogonal polynomials or equivalently MOLS, and in particular, the so related to $e-$Klenian polynomials,  Theorem \ref{importante}.  We believe that this might be improved in order to obtain a complete set of MOLS. Finally, it would be interesting to study the result for fields of characteristic two.

\bibliographystyle{elsarticle-harv}

\begin{thebibliography}{30}

\bibitem{AKT} N. Anbar, C. Kasikci, A. Topuzoglu. On components of vectorial permutations of $\F_q^n$, Finite Fields and their applications, 58(2019) 124-132.

\bibitem{DHK} W.S. Diestelkamp, S.G. Hartke, R.H. Kenney, On the degree of local permutation polynomials, J.Comb. Math. Comb. Comput. 50 (2004) 129–140.

\bibitem{KD} Keedwell A.D., Dénes J.: Latin Squares and their Applications. Elsevier, Amsterdam (2015).

\bibitem{LM} C. F. Laywine, G. L. Mullen.  Discrete Mathematics Using Latin Squares, John Wiley $\&$ Sons, 1998.
 
\bibitem{LN} R. Lidl, H. Niederreiter, Finite Fields, 2nd edn., Encyclopedia Math. Appl., vol.20, Cambridge University Press, Cambridge, 1997.


\bibitem{MGFL} L. Mariot, M. Gadouleau, E. Formenti, A. Leporati, 
Mutually orthogonal latin squares based on cellular automata. Des. Codes Cryptogr. 88(2): 391-411 (2020)

\bibitem{Mo} Montgomery D.C.: Design and Analysis of Experiments. Wiley, Hoboken (2017).


\bibitem{M1} G.L. Mullen, Local permutation polynomials over Zp, Fibonacci Q. 18 (1980) 104–108.

\bibitem{M2} G.L. Mullen, Local permutation polynomials in three variables over Zp, Fibonacci Q. 18 (1980) 208–214.

\bibitem{N1} H. Niederreiter, Permutation polynomials in several variables over finite fields, Proc. Jpn. Acad. 46 (1970) 1001–1005.

\bibitem{N2} H. Niederreiter, Orthogonal systems of polynomials in finite fields, Proc. Am. Math. Soc. 28 (1971) 415–422.


\bibitem{S} D. R. Stinson: Combinatorial characterizations of authentication codes. Des. Codes Cryptogr. 2(2), 175--187 (1992).


\bibitem{W} A. Winterhof,  Generalizations of complete mappings of finite fields and some applications. J. Symb. Comput. 64: 42-52 (2014)


\end{thebibliography}

\end{document}